\documentclass[%
 reprint,
 amsmath,amssymb,
 aps,
]{revtex4-2}

\usepackage{braket}
\usepackage{graphicx}
\usepackage{bm}
\usepackage{hyperref}
\usepackage{algorithm}
\usepackage{algpseudocode}
\usepackage{float}
\usepackage{mathtools}
\usepackage{amsthm}
\usepackage{enumitem}
\usepackage{xcolor}
\usepackage[most]{tcolorbox}

\newtheorem{theorem}{Theorem}
\newtheorem{lemma}{Lemma}

\newtheorem{definition}{Definition}

\newtheorem{assumption}{Assumption}
\newtheorem{proposition}{Proposition}
\newtheorem{protocol}{Protocol}

\newcommand{\bits}{\{0,1\}}
\newcommand{\eps}{\varepsilon}

\newcommand{\R}{\mathcal{R}}
\newcommand{\X}{\mathcal{X}}
\newcommand{\Y}{\mathcal{Y}}

\newcommand{\Tr}{\operatorname{Tr}}

\newcommand{\wt}{\operatorname{wt}}
\newcommand{\negl}{\operatorname{negl}}
\newcommand{\norm}[1]{\left\lVert #1 \right\rVert}
\newcommand{\abs}[1]{\left\lvert #1 \right\rvert}

\newcommand{\proj}[1]{\ket{#1}\!\bra{#1}}

\newcommand{\Id}{\mathbb{I}}
\newcommand{\comst}[1]{\psi_c^{#1}}
\newcommand{\verst}[1]{\psi_v^{#1}}
\newcommand{\dtr}{d_{tr}}

\newcommand{\mbraket}[2]{\langle #1|#2\rangle}
\newcommand{\alg}{\texttt{Alg}}

\newcommand{\mina}[1]{\ifmmode{\color{orange}#1}\else\textcolor{orange}{[MD: #1]}\fi}

\newtcolorbox{protocolboxenv}{%
  enhanced,
  breakable,
  colback=white,
  colframe=black!65,
  boxrule=0.6pt,
  arc=1.2mm,
  left=6pt,
  right=6pt,
  top=6pt,
  bottom=6pt,
  before skip=8pt,
  after skip=8pt
}

\begin{document}

\title{Statistically-Secure Bit Commitment and Coin Flipping Protocols Based on Quantum Hardware Assumptions}

\author{Roo Dunnill}
\email{R.Dunnill@sms.ed.ac.uk}
\author{Mina Doosti}
\affiliation{School of Informatics, University of Edinburgh, Quantum Software Lab}

\begin{abstract}
Bit commitment is impossible to achieve with unconditional security, even in quantum cryptography. We show that statistically secure bit commitment, satisfying both hiding and binding, can be constructed from hybrid locked physical unclonable functions (HLPUFs), a hardware primitive that combines classical hardware tokens and quantum communication. Our protocol uses these hardware assumptions in a novel and non-trivial way to achieve the first mistrustful two-party cryptographic protocol based on hybrid hardware modules. We prove statistical hiding and binding under natural assumptions on the HLPUF and using a carefully designed challenge generation algorithm as a subroutine of our bit-commitment protocol. The construction also yields the first hardware-based coin-flipping protocol. Our results suggest a new paradigm for secure two-party cryptography in quantum networks, combining rigorous security guarantees with a concrete route toward practical implementation.
\end{abstract}

\maketitle

\onecolumngrid

\section{Introduction}
Bit commitment is one of the fundamental primitives of modern cryptography, and a basic building block for a wide range of tasks, including coin flipping, zero-knowledge protocols, secure multiparty computation, and general two-party cryptographic constructions~\cite{IntroToModernCrypto,FoundationsOfCryptography}. Informally, bit commitment allows a sender Alice to commit to a bit in such a way that the value is hidden from the receiver Bob until Alice later chooses to reveal it, while at the same time preventing Alice from changing the bit after the commitment is made. In the quantum setting, bit commitment is particularly important for two reasons. First, it is a canonical test case for how much unconditional security quantum information can offer in mistrustful cryptography. Second, it sits at the centre of an extensive web of reductions: positive and negative results for bit commitment often propagate to many other two-party tasks.

Early optimism that quantum mechanics might enable information-theoretically secure bit commitment was overturned by the Mayers--Lo--Chau impossibility theorem. In the plain non-relativistic two-party model, unconditionally secure quantum bit commitment is impossible: any protocol that is sufficiently concealing necessarily admits a coherent cheating strategy for Alice, based on purification and delayed measurement, that breaks binding \cite{mayers1997qbc,lochau1997qbc}. This impossibility result is one of the foundational no-go theorems of quantum cryptography. It sharply distinguishes key distribution, where unconditional security is achievable, from general two-party computation, where mistrustful tasks are much more constrained \cite{IntroToQuantumCrypto}. The no-go theorem also shows that, in order to achieve this functionality, a compromise needs to be made in the security level or setup assumptions. There are several major directions in which the impossibility theorem has been circumvented. One line of work considers weaker or asymmetric notions of security. For example, Aharonov et al. introduced \emph{quantum bit escrow}, a weaker primitive that captures a form of cheat sensitivity rather than full binding and hiding~\cite{aharonov2000quantumbitescrow}. Another important line develops computationally secure quantum commitments. A seminal result of Dumais et al.~\cite{dumais2000qowp} gives a non-interactive quantum commitment scheme that is perfectly concealing and computationally binding under the assumption of quantum one-way permutations. Later works extended this picture and obtained statistically hiding or statistically binding quantum commitments from broader quantum one-way assumptions \cite{koshiba2009shqbc,koshiba2011niqbc,unruh2016collapse}. A different and highly influential route is to enrich the physical model. In relativistic cryptography, Minkowski causality can be used to enforce non-communication constraints between separated agents, and Kent~\cite{kent2012relativisticbc} showed that unconditionally secure relativistic quantum bit commitment is possible in this setting.

Another notable family of results is achieved by employing physical assumptions. An important work in this area is the bounded-quantum-storage model and the noisy-storage model, where limitations on an adversary's quantum memory are assumed. Bit commitment and other two-party primitives become achievable with information-theoretic security in these models~\cite{damgaard2005bqsm,damgaard2008bqsm,wehner2008noisystorage,ng2012experimental}. These works are especially relevant conceptually for us, as they illustrate the fundamental importance of hardware or physical assumptions in circumventing quantum cryptographic no-go's and as new constructive routes for designing protocols.

This paper follows that philosophy and studies bit commitment from a \emph{hardware assumption}, while not just relying on existential results, but takes a constructive and implementation-aware protocol design mindset. In this work, we answer the following central question:
\begin{center}
    \emph{How can we design constructive and implementable statistically-secure bit-commitment schemes that only rely on hardware assumptions?}
\end{center}

We answer the above question by introducing the first statistically secure quantum bit commitment protocol that achieves both hiding and binding under hybrid hardware assumptions. Our construction relies on a specific hardware module for which we explicitly state the required assumptions, and under these assumptions, the protocol achieves arbitrarily small hiding and binding without any additional adversarial restrictions on the mistrustful parties. Beyond its theoretical significance, our result provides a concrete design based on off-the-shelf hardware components and feasible quantum communication infrastructures, offering a practical route to implementing this fundamental cryptographic primitive in quantum networks. 

Our starting point is the emerging landscape of physically unclonable functions (PUFs) and their quantum variants. Classical PUFs are hardware objects whose challenge-response behaviour is difficult to clone or model, and they have been extensively studied as a resource in hardware security~\cite{PhysicalOneWayFunctions,CHES_FPGAIntrinisicPUF,THIS_StrongPUFs,THIS_StateOfTheArt,SSDaTC_DisorderBasedSecurity}. Quantum PUFs (QPUFs) were introduced in~\cite{2021qpufs} as fully quantum variants of PUFs, and that work showed that, unlike in the classical setting, one can formulate provably secure quantum hardware assumptions. However, existing QPUF-based protocols face many practical limitations: in addition to requiring the transmission of high-dimensional quantum states over the channel, they require a \emph{quantum} database to store the QPUF outputs prior to the initiation of the protocol. Hybrid quantum-classical variants, known as Hybrid Locked PUFs (HLPUFs), were proposed in~\cite{2023quantumlock} to combat these limitations by replacing the quantum database with a classical one. An HLPUF combines a classical PUF, which by itself is insufficient for provable security, with a quantum communication layer that upgrades the resulting hardware module to resist powerful quantum adversaries. In this way, hybrid PUFs aim to combine rigorous security guarantees with practical implementability. So far, however, quantum and hybrid PUFs have mostly been studied for functionalities such as authentication~\cite{2021qpufs,2023quantumlock,doosti2021unifiedframeworkquantumunforgeability,goswami25,laurent2026unconditional} and, more recently, key generation and secure storage~\cite{Nilesh2025QPUFStorage}. In all of these settings, the honest parties use the hardware against an external network adversary. By contrast, quantum or hybrid hardware assumptions have not previously been explored for inherently mistrustful two-party tasks such as bit commitment or coin flipping.

This makes the study of mistrustful multiparty protocols in this setting both interesting and technically challenging. A central conceptual insight of our work is that there is a fundamental reason for this difficulty: hardware assumptions fit naturally into settings with honest users and an external adversary, but they do not immediately extend to mistrustful two-party cryptography. Indeed, as we explain in the next section, a naive attempt to use a hardware module for bit commitment is bound to fail for a simple but crucial reason: if one of the parties is malicious, how can the protocol force that party to actually use the prescribed hardware token at all? Our protocol is designed specifically with this limitation in mind. It combines an enrolled and later locked hardware token with a BB84-style encoding of a challenge-dependent response substring, and it is carefully structured so that replacing the token with a lookalike device becomes statistically difficult. The key idea is to introduce a deliberate asymmetry between the verification part of the response and the committed stream of states. This asymmetry is essential both for ruling out the main cheating strategy and for evading the standard quantum no-go result.

We also present a natural coin-flipping protocol built black-box from our bit-commitment scheme. To the best of our knowledge, this is also the first strong quantum coin-flipping protocol based on hybrid hardware assumptions. Combined, our results suggest a new paradigm for secure two-party cryptography in quantum networks, combining rigorous security guarantees with a concrete route toward practical implementation.

\subsection{Our technical contribution}
To elaborate on our technical contribution, we first note that a naive hardware-based bit-commitment design is insufficient. If Alice is simply asked to prepare a commitment using a hardware token and later open it with a classical challenge-response pair, there is no reason a malicious Alice must actually use the prescribed token in the intended way. In particular, if the protocol only checks consistency of the opening data after the fact, then Alice may try to delay the effective choice of the bit until the opening phase, or replace the honest hardware by a look-alike strategy whose behaviour is driven only by the revealed opening data and not by the challenge-dependent rule that the real hardware is supposed to enforce. The main obstacle is therefore not merely to hide the committed value, but to force the opening transcript to be anchored to genuine hardware behaviour while still keeping the bit hidden before the reveal stage.

Our first contribution is to show how to overcome exactly this obstacle using a deliberately asymmetric variant of an HLPUF. At a high level, the hardware response is split into two parts: a shorter verifier part and a longer payload part. The verifier part is used only to unlock and authenticate the hardware behaviour for a given challenge, while the payload part is what is encoded into the commitment state. The commitment itself is formed by taking the same payload and encoding it in one of two BB84 basis patterns determined by a complete challenge response pair and a carefully generated alternative challenge. The pair of candidate challenges is revealed already in the commit phase, but the hidden information is which of the two corresponding basis patterns was used to encode the commitment. This asymmetry is what makes the construction work: it allows the commitment state to conceal the bit, while ensuring that a successful opening must still be tied to a genuine challenge-response pair embedded in honest hardware.

Our second contribution is a statistical hiding analysis for this construction. Informally, if the enrolled substrings are close to uniform, then the two honest commit states are correspondingly close, and in the ideal case, they are perfectly indistinguishable. In particular, the protocol achieves hiding because Bob receives a BB84-encoded commitment of a hardware response under two possible candidate challenges, which before opening, he does not know which basis pattern is the relevant one. At the informal level, the guarantee can be summarised as the trace distance between the commitment states $d_{\mathrm{tr}}(\rho_0,\rho_1)\le \eta,$ where $\eta$ measures the deviation of the target distribution from uniform. In the ideal uniform case, this gives perfect hiding. At the proof level, the key idea is an intrinsic ensemble-symmetry argument: the pair-challenge generation procedure is designed so that, for all sent qubits, the four BB84 states appear in a statistically balanced way. Consequently, the two honest commitment ensembles are locally identical on these coordinates and hence globally identical, up to the residual leakage caused by the possibility that Bob can forge access to the locked hardware.

Our third contribution is a binding proof that captures all possible attack strategies. The binding argument is based on finding the operator norm of the sum of all accepted openings. The core of the proof relies on the input-output spaces of the HLPUF, preventing the standard entanglement based bit commitment attacks to occur. This, combined with the commitment lying in the pre-disclosed basis preparations results in a bound of
\begin{equation}
    p_0+p_1\le1+2^{\frac{-\ell_{min}}{4}}
\end{equation}
Where $p_b$ is the probability that a cheating Alice successfully opens bit $b$. 

Beyond the security statements themselves, we view the proof techniques as one of the main contributions of the paper. On the hiding side, the central idea is to impose a balanced combinatorial structure on the generated challenges, which makes the honest commitment ensembles intrinsically symmetric. On the binding side, the proof reduces arbitrary cheating strategies to an operator-norm bound for a small family of acceptance projectors. This reduction cleanly separates the intrinsic quantum-overlap limitation from the hardware-dependent parameters. Together, these two ingredients provide a constructive route to statistically secure bit commitment from hybrid hardware assumptions, and they also lead directly to a coin-flipping protocol via the standard commitment-based transformation.

\section{Preliminaries and notations}~\label{sec:prelim}
We start by introducing some of our main notations, as well as background on primitives and definitions that we use throughout the paper. \\

For $\beta\in\bits$, let $\ket{z}_{\beta}$ denote the single-qubit BB84 encoding of $z\in\bits$ in basis $\beta$, where basis $0$ is the computational basis (or Z basis) and basis $1$ is the Hadamard basis (or X basis). For strings $z,\beta\in\bits^{\ell_{min}}$, we define
\begin{equation}
    \ket{z}_{\beta} := \bigotimes_{j=1}^{\ell_{min}}\ket{z_j}_{\beta_j}.
\end{equation}
If $a\in\bits^{2k}$, we write $a=(\theta,u)$ with $\theta,u\in\bits^k$ and define the $k$-qubit BB84 state where $u$ denotes the value bit and $\theta$ denotes the basis bit. (ex: for $a=(1,0)$ would correspond to state $\ket{\psi^{a}}:=\ket{u}_{\theta} = \ket{+}$).

\subsection{Distance measures}
We introduce the notation and definitions for the classical and quantum distances that we use. 
Let $x$ and $y$ be sets of $x,y\in\mathbb{F}_2^{n}$. we denote their Hamming distance as follows:
    \begin{equation}
            d(x,y)=|\{i\in\{1,...,n\}|x_i\neq y_i\}|
    \end{equation}
We also use the Hamming weight to measure the number of non-zero elements of a group. Let $x$ be a set $x\in\mathbb{F}_2^{n}$. Hamming weight is defined as:   
    \begin{equation}
        \wt(x)=|\{i\in\{0,...,n\}|x_i\neq0\}|
    \end{equation}
For two quantum states $\rho,\sigma\in(\mathcal{H}^2)^n$, the trace distance is:   
\begin{equation}
    \dtr (\rho,\sigma):=\frac12\norm{\rho-\sigma}_1.
\end{equation}
We also use the operator norm. For a linear operator $A$, its operator (infinity) norm is
\begin{equation}
\|A\|_\infty := \sup_{\|x\|_2=1} \|Ax\|_2.
\end{equation}
Equivalently, $\|A\|_\infty = \sigma_{\max}(A),$, where $\sigma_{\max}(A)$ denotes the largest singular value of $A$.

\subsection{Bit commitment}
Bit commitment is a central two-party cryptographic primitive that enables one of the parties to commit a bit value and reveal it later to the other party. The two core security features are that the committing party cannot change their commitment after they have sent it to the party, and that the receiving party cannot open this commitment before the reveal step. This is known as the \emph{hiding (concealing)} and \emph{binding} property of a commitment scheme.\\

\noindent \textbf{Bit Commitment.} \textit{A bit-commitment protocol between Alice and Bob has two phases.
\begin{enumerate}[label=(\roman*)]
    \item In the commit phase, Alice with input $b\in\bits$ interacts with Bob and leaves Bob with a commitment transcript and possibly a quantum register.
    \item In the open phase, Alice sends opening information and Bob outputs either a bit $b'\in\bits$ or reject.
\end{enumerate}
The protocol is \emph{correct} if an honest execution with committed bit $b$ is accepted with probability $1$ and opens to $b$.}\\

In a dishonest scenario, the bit-commitment is expected to have two properties known as \emph{hiding (or concealing)} and \emph{binding}, defined as below:

\begin{definition}[$\epsilon$-hiding Bit Commitment]
    A quantum bit commitment scheme is $\epsilon$-hiding (concealing) if the receiving party satisfies:
    \begin{equation}
        p_0+p_1-1\leq\epsilon(n)
    \end{equation}
    Where $p_b$ is the probability the receiver can guess the commitment $b$ correctly.
\\
Equivalently (more relevant for quantum commitment schemes), a bit-commitment protocol is $\eps$-hiding if for every dishonest Bob the joint states after the commit phase corresponding to honest commitments of $0$ and $1$ is at most $\eps$-distinguishable.
\end{definition}

\begin{definition}[$\epsilon$-binding Bit Commitment]
    A quantum bit commitment scheme is $\epsilon$ binding if the committing party satisfies:
    \begin{equation}
        p_0+p_1-1\leq\epsilon(n)
    \end{equation}
    Where $p_b$ is the probability the committer can reveal $b$.
\end{definition}

This notion of binding is often referred to as \emph{weak-binding}, and it is the formulation often used for quantum bit commitment schemes~\cite{mayers1997}.

\subsection{Coin flipping}
Coin flipping is another fundamental two-party primitive, first introduced by Blum\cite{blum1983CoinFlippingByTelephone} and allows two mistrustful parties to agree on a shared random bit. The security of the protocol relies on neither of the parties being able to bias the protocol with probability higher than $\frac{1}{2}+\epsilon$. Strong coin flipping is a variant of coin flipping where both outcomes are restricted by this bias.

\begin{definition}[Strong coin flipping]
A two-party protocol outputs a bit $c\in\bits$ or aborts. It is correct if, when both parties are honest, the output is uniformly distributed on $\bits$ and the protocol does not abort. It has bias at most $\epsilon$ if for every dishonest Alice and every dishonest Bob,
\begin{equation}
\Pr[c=c^*] \le \frac12+\epsilon
\end{equation}
for each target bit $c^*\in\bits$, except with abort probability explicitly accounted for by the protocol.
\end{definition}

\subsection{Classical and Hybrid PUFs}
Physical Unclonable Functions (PUFs) are hardware primitives that use the intrinsic randomness of physical devices as unique fingerprints, or physical keys, with the idea of providing unpredictability and unclonability against external adversaries and even the manufacturer of the device~\cite{galetsky22}. As a formal cryptographic primitive, some PUFs give a physical instantiation of a one-way function through hardware assumptions~\cite{PhysicalOneWayFunctions}. Classical PUFs (cPUFs) can be modelled as $f: \{0,1\}^n \rightarrow \{0,1\}^m$. However, the description of the function itself is not accessible due to the hardware assumption and the behaviours of these devices are characterized via sets of Challenge-Response Pairs (CPRs) $\{(x_i,r_i)\}^q_i$ through querying the device. 

A classical PUF can also be modelled as a probabilistic function $f:\R \times\X\rightarrow\Y$ where $\X$ is the input space, $\Y$ is the output space of $f$ and $\R$ is the identifier. The creation of a classical PUF is formally expressed by invoking a manufacturing process $f\leftarrow\mathcal{MP}_{C}(\lambda)$, where $\lambda$ is the security parameter. An important property of a classical PUF in this model is the notion of $\emph{randomness}$, which is the maximal probability of $p^f_x(y)$ with an input $x_j\in \X$ on PUF $f_i$ where $i\in\R$. 

\begin{definition}[$p$-Randomness (from~\cite{2023quantumlock})]\label{def:p-randomness}
We define the $p$-randomness of a classical PUF $f:\R \times \X\rightarrow\Y$ as
\begin{equation}
    p := \max_{\substack{x \in \X \\ y \in \Y}} p^f_x(y).
\end{equation}
\end{definition}

Apart from these, a device with an underlying function $f$ needs to satisfy certain requirements to be qualified as a PUF. There exist different sets of such requirements in the literature~\cite{armknecht16, 2021qpufs, doosti2021unifiedframeworkquantumunforgeability, 2023quantumlock}, however, for this work we adopt the minimal set of requirements also used in hybrid schemes such as in~\cite{2023quantumlock}, which includes: $\delta_{1}$-Robustness, $\delta_{2}$-Collision Resistance, and $\delta_{3}$-Uniqueness. We omit the formal definition of these requirements for the purpose of this work, as we are not explicitly using them here, and whenever referring to a classical PUF, we assume these requirements are satisfied.

\subsubsection{Hybrid and Hybrid-locked PUFs}
The main hardware module we use for our protocol is a specific type of Hybrid Locked PUF (HLPUF) introduced in \cite{2023quantumlock}. An HLPUF is a hybrid device that combines a cPUF with quantum encoding and verification. Given the vulnerabilities of most existing cPUFs against modelling and machine learning attacks, HLPUFs are designed to boost the security of cPUFs against both general classical and also powerful quantum adversaries. The construction first takes a cPUF $f:\{0,1\}^n\rightarrow\{0,1\}^{2m}$ and encodes the outputted bitstring into BB84 states. The output is split into $m$ pairs of bits where:
\begin{equation}
\forall\,1\le j\le m,\qquad
f(x)_{2j-1}\,f(x)_{2j}\mapsto
\left\{
\begin{array}{ll}
00 \mapsto \ket{0}, & 01 \mapsto \ket{1},\\
10 \mapsto \ket{+}, & 11 \mapsto \ket{-}.
\end{array}
\right.
\end{equation}
This creates the mapping $\mathcal{E}_f:\{0,1\}^n\rightarrow(\mathcal{H}^2)^m$ where $x\rightarrow\otimes^{m}_{j=1}\ket{\psi_{2j-1,2j}}\bra{\psi_{2j-1,2j}}$,
which is referred to as a HPUF. This encoding alone significantly improves security against powerful modelling attacks, but it does not protect against adaptive adversaries. To address this, an additional mechanism, called \emph{lock}, is incorporated into the construction to provide security against adaptive quantum adversaries. Its output is partitioned as \(f(x)=f_1(x)\|f_2(x)\), where only \(f_2(x)\) is encoded into BB84 states by the HPUF. The lock component takes as input both the challenge and an \(m\)-qubit state, and outputs the state \(\ket{\psi_{f_2(x)}}\) only if the supplied quantum state passes an internal verification procedure. In this verification step, the device measures the \(j\)-th qubit in the basis specified by the \((2j-1)\)-st bit of \(f_1(x)\), where \(0\) denotes the computational basis and \(1\) denotes the Hadamard basis. The device accepts only if, for every \(1\le j\le m\), the measurement outcome equals the \(2j\)-th bit of \(f_1(x)\). This locking mechanism prevents adaptive querying of the cPUF and thereby provides unforgeability against such adversaries. This is particularly important in adversarial protocols, where the party holding the device may otherwise have adaptive access to it.

These devices were originally introduced for authentication, with security based on a collection of hardware assumptions. The central assumption is that the device cannot be opened or tampered with and is therefore treated as a black box. This complements the locking mechanism: the HLPUF may initially operate in an unlocked mode, in which adaptive queries are possible, but once it is locked, it is assumed to be impossible to unlock it again at any later stage.

\section{Bit-commitment protocol}
At a high level, our protocol begins with the committer, Alice, holding the device in its unlocked mode and using it to create a classical database. She then locks the device and sends it to the verifier, Bob, who will later use it for verification. The main subtlety of the protocol is to exploit the properties of a specific HLPUF variant in a way that Alice is forced to use the hardware module to generate both the commitment state and the associated transcript needed to open the commitment later. At the same time, Bob cannot use the locked device to learn or open the committed value before the opening phase, due to the unforgeability properties of the HLPUF. Conceptually, hiding follows directly from the hardware assumption, whereas binding is more subtle and depends on the specific structure and properties of the protocol.

In our protocol, Alice’s commitment to a bit $b$ is a quantum state consisting of a sequence of BB84 states that encodes a partial response of the HLPUF under a challenge that serves as a key. For each commitment state, there are also two associated challenges, $x_0$ and $x_1$. The challenge $x_0$ supplies the response substring used in the commitment, while $x_1$ is generated carefully by a dedicated algorithm and supplies the alternative basis pattern. Both are valid challenges of the HLPUF. These two challenges are announced publicly to Bob, but the crucial hidden information is which basis pattern was used for the encoding. In the opening phase, the full challenge-response relation is revealed, including the used challenge and the two halves of its response, $f_1(x_0)$ and $f_2(x_0)$. This information enables Bob to query the HLPUF in locked mode and verify that Alice did not cheat.

To present our bit-commitment protocol, we first introduce two main ingredients. The first is the construction of the hardware module used in the protocol, together with the assumptions on which its security relies. The second is the algorithm used to generate alternative challenges. We begin with the hardware construction and its assumptions.

\subsection{Asymmetric HLPUF construction and assumptions}\label{sec:const}
Let us first specify the construction of the Hybrid Locked PUF that we need for our protocol, together with hardware assumptions. Our construction is similar to the one in~\cite{2023quantumlock}, using an underlying classical PUF $f:\{0,1\}^n\rightarrow\bits^m$, but with an asymmetric split in the response of the underlying classical PUF. As such, we refer to our variant here as an \emph{asymmetric HLPUF}. Concretely, the classical response is split as
\begin{equation}
f(x)=f_1(x)\,\|\,f_2(x),
\qquad
f_1:\mathcal{X}\to\bits^{s},
\qquad
f_2:\mathcal{X}\to\bits^{t},
\qquad s+t=m,
\end{equation}
where $s=2k$ and $t=2l$ are even.

For a string $z\in\bits^{2l}$, write
\begin{equation}
z=(\theta,u),
\qquad
\theta,u\in\bits^l,
\end{equation}
and define the associated $l$-qubit BB84 state
\begin{equation}
\ket{\comst{z}}:=\ket{u}_{\theta}.
\end{equation}

In the unlocked mode (mode 0), the asymmetric HLPUF takes as input $x$ and outputs the full classical response
\begin{equation}
f(x)=f_1(x)\,\|\,f_2(x).
\end{equation}
In the locked mode (mode 1), the asymmetric HLPUF takes as input a challenge $x$ and a $k$-qubit state $\rho$, performs the verification measurement
\begin{equation}
\bigl\{\proj{\verst{f_1(x)}},\, \Id-\proj{\verst{f_1(x)}}\bigr\},
\end{equation}
and, if the first outcome occurs, outputs the $l$-qubit state
\begin{equation}
\ket{\comst{f_2(x)}}.
\end{equation}
Otherwise it aborts and outputs $\bot$.

\begin{assumption}[Locked-mode black-box property and irreversibility ]\label{ass:hardware-module}
In the locked mode (mode 1), it is impossible to obtain direct classical access to $f_2(x)$ from the hardware token. In particular, the locked token can only be accessed through the interface above, which returns either the quantum state $\ket{\comst{f_2(x)}}$ or $\bot$.
Also, once the token has been switched from mode 0 to mode 1, it cannot be unlocked again, and the functions $f_1$ and $f_2$ embedded in the token cannot be altered or queried except through the public locked-mode interface.
\end{assumption}

\begin{assumption}[cPUF bias]\label{ass:HLPUF}
The underlying classical PUF is imperfect and can have a $\delta_r$ bias. Namely, for every output bit $j\in[m]$ and every distribution on challenges,
\begin{equation}
\left|\Pr[f(X)_j=1]-\frac12\right|\le \delta_r.
\end{equation}
Moreover, the output bits are independent and identically distributed.
\end{assumption}

\begin{assumption}[HLPUF unforgeability]\label{ass:unf}
The Asymmetric HLPUF is $\eps_{\mathrm{forge}}$-unforgeable where $\eps_{\mathrm{forge}}(s,t)$ is a negligible function in $s$ and $t$.
\end{assumption}

Unforgeability is the main security property of PUFs and refers to the inability of an adversary to predict or produce valid outcomes for unseen instances. We rely on the general unforgeability of HLPUFs established in~\cite{2023quantumlock}, and we assume that this property is satisfied by the construction we use here. Since the proof in our setting follows as a subset of the arguments already given in~\cite{2023quantumlock}, we do not repeat it here.

\subsection{Challenge-pair generation algorithm}
As mentioned before, our protocol uses a subroutine whose role is to generate an alternative challenge $x_1$ from $x_0$ such that they satisfy certain properties. For a given randomly selected challenge $x_0$, it flips the $Jth$ bits of $x_0$ to create $x_1$. The set $J$ is designed in a way so that for a given random string $r\in\{0,1\}^{\ell_{min}}$, the corresponding $f_2(x_0)_J=r$ where $\ell_{min}:=\alpha t$ with $0<\alpha\le1$. Finally it uses oracle access to the unlocked HLPUF to test to see if their response provides enough separation to be used for verification. It outputs the first candidate that passes all filters. Here we need to make a note about the notation first. Each challenge \(x\in\bits^n\) is written as
\begin{equation}
x=(x^{\mathrm{bas}},x^{\mathrm{aux}}),
\qquad
x^{\mathrm{bas}}\in\bits^{t},\quad x^{\mathrm{aux}}\in\bits^{n-t}.
\end{equation}
where the $t$-bit string $x^{\mathrm{bas}}$ determines the generation of an alternative challenge in the algorithm below. Whenever Algorithm~\ref{alg:fake} selects coordinates, we write $J=(j_1,\ldots,j_{\ell_{\min}})$ for the ordered list of selected indices; subscripts such as $f_2(x_0)_J$ follow this order. When only membership or Hamming distance is relevant, we identify $J$ with its underlying set. Since we mostly care about these basis value, we denote 
\begin{equation}
    \beta(x):=x^{\mathrm{bas}}_J
\end{equation}
Where $|J|=\ell_{min}$ and we use this notation instead of $x$ inside the algorithm. The remaining $n-\ell_{min}$ bits play no direct role in the encoding and serve only as auxiliary challenge degrees of freedom in the alternative challenge generation procedure. In particular, for any pair of challenges $(x_0,x_1)$, the number of coordinates on which the two basis strings differ is $d\!\left(\beta(x_0),\beta(x_1)\right)$ which is exactly the Hamming distance between $\ell_{min}$ bits of the two challenges.

To minimise the probability of breaking the binding while obtaining the bound stated above, we choose the length of the verification string as $s:=\frac{\ell_{min}}{4}$. The security parameter $\ell_{min}$ should be chosen between $\frac{\alpha}{2}\leq p \leq 1-\frac{\alpha}{2}$ where $p:=\frac{1}{2}+\delta_r$. The reasoning for these choices are given in Theorem \ref{thm:combinatorial-feasibility-ext} and Lemma \ref{lem:sparse-verifier}. This algorithm also ensures that the distribution of encoded bits of $f_2(x_0)$ are chosen randomly while the basis they are encoded in are also uniformly randomly sampled; allowing for the indistinguishability condition to be met. While Bob receives both challenges in the commit phase, this indistinguishability condition results in him gaining no information on the bit committed.

The formal description of the Algorithm is given below.
\begin{algorithm}[H]
\caption{Balanced alternative-challenge generation}
\label{alg:fake}
\begin{algorithmic}[1]
\Require a challenge $x_0$, oracle access to the asymmetric HLPUF in unlocked mode, and a target parameter \(\ell_{\min}\)
\Ensure either a challenge $x_1$ together with the selected ordering $J$ or \textbf{fail}
\State Choose a random string $r\overset{\$}{\leftarrow}\{0,1\}^{\ell_{min}}$
\State Query the unlocked HLPUF on $x_0$ and obtain
\[
f(x_0)=f_1(x_0)\,\|\,f_2(x_0),
\qquad
f_2(x_0)=u\in\bits^{t}.
\]
\If{$\wt(f_2(x_0))\geq\wt(r)\;\;\wedge\;\;t-\wt(f_2(x_0))\geq \ell_{min}-\wt(r)$}
\Else
    \State \textbf{return fail}
\EndIf
\Repeat
    \State Initialise $J\gets()$
    \For{$i\in[|r|]$}
        \State Sample $j \overset{\$}{\leftarrow} \{k\in[t]:f_2(x_0)_k=r_i\} \setminus J$
        \State $J\gets(J,j)$
    \EndFor
    \State $x_1 \gets x_0$
    \State $(x_1)_j \gets (x_0)_j \oplus 1,\;\;\forall\;j \in J$
    \State Query the unlocked HLPUF on $x_1$ and obtain $f_1(x_1)$.
\Until$\abs{\mbraket{\psi^{f_1(x_0)}}{\psi^{f_1(x_1)}}}^2\le 2^{-s/2}$
\State \textbf{return} $x_1$ and retain the selected ordering $J$
\end{algorithmic}
\end{algorithm}
Algorithm~\ref{alg:fake}, which we also denote as $\alg(x_0)$ for brevity (referring to its returned challenge $x_1$, with $J$ retained as auxiliary information), has certain properties that we use in the proofs, and we list them here to clarify their importance.

\begin{proposition}[Properties of Algorithm~\ref{alg:fake}]
\label{prop:fake-properties}
Suppose Algorithm~\ref{alg:fake} returns another challenge $x_1$ for the original challenge $x_0$. Let $J=(j_1,\ldots,j_{\ell_{\min}})$ be the ordered list retained by the algorithm; its underlying set is exactly
\begin{equation}
    \{\,j\in[t]:(x_0)_j\oplus (x_1)_j=1\,\}.
\end{equation}
Then the algorithm has the following properties:
\begin{enumerate}[label=(\roman*)]
    \item \textbf{Challenge permutability:} both challenges have the same probability of producing the other when entered into the algorithm
    \begin{equation}
        \Pr[\alg(x_0)=x_1]=\Pr[\alg(x_1)=x_0]
    \end{equation}
    \item \textbf{Challenge Hamming weight probability:} given an inputted randomly chosen challenge, both challenges have the same probability of Hamming weights
    \begin{equation}
    \forall\;i\;\Pr[(x_0)_i=1]=\Pr[(x_1)_i=1]=\frac{1}{2}
    \end{equation}

    \item \textbf{Large basis-distance:} The Hamming-distance of the basis bits of both challenges is equal to a certain threshold $l_{\min}$
    \begin{equation}
    d\!\left(\beta(x_0),\beta(x_1)\right)=\ell_{\min},
    \end{equation}
    Equivalently,
    \begin{equation}
    |J|
    =
    d\!\left(\beta(x_0),\beta(x_1)\right)=\ell_{\min},
    \end{equation}
    \item \textbf{Perfect value and basis balancing:} On the differing coordinates, the basis and value bits are balanced such that they produce evenly distributed BB84 states. With the ordered list $J=(j_1,\ldots,j_{\ell_{\min}})$ selected by the algorithm, $u_J=f_2(x_0)_J=r$, and hence the selected value string is uniformly distributed:
    \begin{equation}
        \Pr[u_J=v]=2^{-\ell_{\min}}\qquad\forall\,v\in\bits^{\ell_{\min}}.
    \end{equation}
    In particular, $\forall\;j\in J$
    \begin{equation}
    \Pr[u_j=1]=\frac{1}{2}
    \end{equation}
    and
    \begin{equation}
    \Pr[(x_0)_j=1]=\frac{1}{2}
    \end{equation}

    \item \textbf{Verifier separation:} the encoded state of the $f_1(.)$ for the two challenges have small overlap.
    \begin{equation}
    \abs{\mbraket{\psi^{f_1(x_0)}}{\psi^{f_1(x_1)}}}^2\le 2^{-s/2}.
    \end{equation}
\end{enumerate}
\end{proposition}
All of these properties are straightforwardly satisfied by the specific steps of the algorithm. However, for the algorithm to be useful in our protocol, an additional important requirement is that it can efficiently find an alternative candidate for a given selected challenge. In other words, the algorithm must be efficient both in its query complexity to the HLPUF and in its combinatorial complexity. In the following three theorems, we show that this is indeed the case. The proofs are deferred to the Appendix~\ref{app:alg-complexity}.
\begin{lemma}[Symmetry of Challenge Generation]
    Let $x_0\in\{0,1\}^n$ be a challenge selected uniformly at random. Let $x_1$ be a challenge created by Algorithm \ref{alg:fake} from $x_0$, by flipping a set of indices in $J\subseteq[t].$ Given the unforgeability conditions of the HLPUF, then:
    \begin{equation}        \Pr\left[\alg(x_0)=x_1\right]=\Pr\left[\alg(x_1)=x_0\right]
    \end{equation}
    \begin{proof}
        Let $J=\{j_1...j_i\}\subseteq[t]$ be the set of indices selected in \ref{alg:fake}. The algorithm produces $x_1$ from $x_0$ by flipping all indices within set $J$.
        \begin{equation}
            \forall\;j\in J,\;\;\;(x_1)_j\gets(x_0)_j\oplus1
        \end{equation}
        By the assumptions of the construction of the HLPUF, the output distribution of $f_2(x)$ is identical for all $x\in\{0,1\}^n$. As such, for any two challenges $x,x'$, the distributions of $f_2(x), f_2(x')$ are identical. Consequently, the distribution of all possible sets $J$ generated from a response is identical for all uniformly sampled challenges. Since the indices are selected at random from the positions of $f_2(x)$, the probability of selecting a set $J$ is identical for all challenges. As such, the probability to obtain a specific set $J^*:=\{j\in[t]:(x_0)_j\neq(x_1)_j\}$ is the same for both challenges $x_0,x_1$. Furthermore, XOR is self-inverse. As such, applying $J^*$ to $x_0$ generates $x_1$ and applying $J^*$ to $x_1$ generates $x_0$. Thus the probability that $x_0$ generates $x_1$ is the same as the probability that $x_1$ generates $x_0$
        \begin{equation}
            \Pr[\alg(x_0)=x_1]=\Pr[J=J^*|x_0]=\Pr[J=J^*|x_1]=\Pr[\alg(x_1)=x_0]
        \end{equation}
    \end{proof}
\end{lemma}
\begin{theorem}[High-Probability  Combinatorial Feasibility Test in Algorithm~\ref{alg:fake}]
\label{thm:prob-suc-alg1}
Let $\ell_{min}$ be the number of qubits in $\rho_b$, let $t$ be the length of the string $f_2(x_0)$ and let $\delta_r$ be the bias of the HLPUF. The probability that the combinatorial feasibility test in Algorithm~\ref{alg:fake} is satisfied for one sampled string $r$ is:
\begin{equation}
    \Pr[success]=\sum_{y=0}^{\ell_{min}}\left(\frac{1}{2}\right)^{\ell_{min}}\binom{\ell_{min}}{y}\sum_{k=y}^{t-(\ell_{min}-y)}\binom{t}{k}\left(\frac{1}{2}+\delta_r\right)^k\left(\frac{1}{2}-\delta_r\right)^{t-k}
\end{equation}
\end{theorem}
\begin{theorem}[High-probability combinatorial feasibility of Algorithm~\ref{alg:fake}]
\label{thm:combinatorial-feasibility}
Let $x_0 \xleftarrow{\$} \{0,1\}^{n}$, and suppose that the bits of
$f_2(x_0)\in\{0,1\}^{t}$ are independently and identically distributed with a bias at most $\delta_r$. where $p:=\Pr[f_2(x_0)_j=1]$. Let $r\xleftarrow{\$}\{0,1\}^{\ell_{\min}}$ be sampled independently,
where
\begin{equation}
    \ell_{\min}=\alpha t,\qquad 0<\alpha<1,
\end{equation}
and assume that $\ell_{\min}$ is an integer. If
\begin{equation}
    \frac{\alpha}{2}<p<1-\frac{\alpha}{2},
\end{equation}
then the combinatorial feasibility condition in Step~3 of Algorithm~1,
namely
\begin{equation}
    \operatorname{wt}(f_2(x_0))\geq \operatorname{wt}(r)\quad\text{and}\quad t-\operatorname{wt}(f_2(x_0))\geq\ell_{\min}-\operatorname{wt}(r),
\end{equation}
is satisfied with probability
\begin{equation}
    \Pr[\mathrm{feasible}]
    =
    1-e^{-\Omega(t)}.
\end{equation}
\end{theorem}

\begin{theorem}[Query complexity of Algorithm~\ref{alg:fake}]
\label{thm:fake-query}
Conditioned on the combinatorial feasibility test being satisfied, for every fixed $\delta_r<1/2$ the expected number of unlocked-HLPUF queries made by Algorithm~\ref{alg:fake} satisfies
\begin{equation}
    \mathbb{E}[Q]=2+e^{-\Omega(s)}.
\end{equation}
\end{theorem}

\begin{protocolboxenv}
\begin{protocol}[HLPUF bit commitment.]\label{bc-prot}

\medskip
\textit{Public parameters:} security parameters $s=2k$ and $t$; challenge space $\X$; minimum basis-distance parameter $\ell_{\min}$; The protocol uses an Asymmetric HLPUF as defined in Section~\ref{sec:const} with bias parameter $\delta_r$.

\medskip
\textbf{Setup phase:}
\begin{enumerate}[label=\arabic*.]
    \item Alice queries the HLPUF in \texttt{mode 0} on enough challenges to obtain a database of challenge-response pairs
    \[
    DB=\{(x,f_1(x),f_2(x)) : x\in\mathcal{U}\}.
    \]
    For each enrolled challenge \(x\in\mathcal{U}\), she also runs Algorithm~\ref{alg:fake} to produce an associated challenge $x_1$ and retains the corresponding ordered list $J$.
    \item Alice locks the HLPUF, switches it to \texttt{mode 1}, and sends it to Bob.
\end{enumerate}

\textbf{Commit phase:}
The commitment is a bit \(b\in\bits\).
\begin{enumerate}[label=\arabic*.]
    \setcounter{enumi}{2}
    \item Alice samples a challenge $x_0\in\mathcal{U}$, together with its associated challenge $x_1$ and ordered list $J$.
    \item Alice sends Bob the pair $(x_0, x_1)$ together with the ordering $J$.
    \item Alice prepares the $\ell_{min}$-qubit state
    \[
    \rho_b :=
    \begin{cases}
    \proj{f_2(x_0)_J}_{\beta(x_0)} & \text{if } b=0,\\[1mm]
    \proj{f_2(x_0)_J}_{\beta(x_1)} & \text{if } b=1.
    \end{cases}
    \]
    Where $f_2(x_0)_J$ corresponds to the indices of $f_2(x_0)$ that are in $J:=\{j\in[t]:(x_0)_j\neq(x_1)_j\}$ and $|J|=\ell_{min}$.
    \item Alice sends the register containing \(\rho_b\) to Bob.
    \item Bob checks that $\forall j\in J, (x_0)_j\oplus(x_1)_j=1$, otherwise he rejects
\end{enumerate}

\textbf{Open phase:}
\begin{enumerate}[label=\arabic*.]
    \setcounter{enumi}{6}
    \item Alice sends the triple \((b,a,z)\), where
    \[
    a=f_1(x_0),
    \qquad
    z=f_2(x_0).
    \]
\end{enumerate}

\textbf{Verification phase:}
\begin{enumerate}[label=\arabic*.]
    \setcounter{enumi}{7}
    \item Bob prepares the verifier state \(\ket{\verst{a}}\) from the revealed classical string \(a\).
    \item Bob queries the HLPUF on \((x_0,\ket{\verst{a}})\). If the token outputs \(\bot\), he rejects. Otherwise he receives a $l$-qubit state \(\sigma\).
    \item Bob measures \(\sigma\) with the two-outcome test
    \[
    \{\proj{\comst{z}},\,\Id-\proj{\comst{z}}\}.
    \]
    If this test rejects, he rejects the opening.
    \item Bob measures the stored commitment register in basis
    \[
    \gamma_b:=
    \begin{cases}
    \beta(x_0) & \text{if } b=0,\\
    \beta(x_1) & \text{if } b=1.
    \end{cases}
    \]
    obtaining a string \(z'\in\bits^{\ell_{min}}\). If \(z'\neq z_J\), he rejects. Otherwise, he accepts bit \(b\).
\end{enumerate}
\end{protocol}
\end{protocolboxenv}

\subsection{Formal protocol description}
Our hardware-based bit-commitment protocol works as follows: In the setup phase, Alice uses the HLPUF in its unlocked mode to build a database of valid challenge-response pairs and, for each challenge, computes a corresponding challenge. She then locks the device and sends it to Bob. To commit to a bit $b$, Alice chooses a challenge $x_0$ together with its counterpart $x_1$, sends Bob the ordered pair $(x_0,x_1)$ together with the retained ordering $J$, and prepares a BB84 sequence of states using the response $f_2(x_0)$, encoded either in the basis of the challenge $x_0$ or in the basis of the alternative challenge $x_1$ depending on the committed bit. Later, to open the commitment, Alice reveals which challenge was used to prepare the state and provides the corresponding classical response information. Bob then uses the locked HLPUF to verify that the revealed response is genuine, and checks that the previously received quantum state is consistent with the claimed bit. In this way, the protocol forces Alice to commit using valid HLPUF data, while preventing Bob from learning the bit before the opening phase. The formal description of the protocol is given in Protocol~\ref{bc-prot}

\section{Security proof}
We now formally prove correctness, hiding and binding for Protocol~\ref{bc-prot}.
\begin{theorem}[Correctness]
\label{thm:correctness}
If both parties are honest and the announced pair $(x_0,x_1)$ was produced by Algorithm~\ref{alg:fake}, then Protocol~\ref{bc-prot} opens correctly to the committed bit, and Bob accepts with probability $1$.
\end{theorem}

\begin{proof}
Let Alice commit to bit $b\in\bits$, and let
\begin{equation}
z=f_2(x_0)\in\bits^t,
\qquad
a=f_1(x_0)\in\bits^s.
\end{equation}
Let $z_{J}$ be the elements of $z$ in the group $J:=\{j\in[t]:(x_0)_j\neq(x_1)_j\}$ and $|J|=\ell_{min}$. By the protocol, Alice sends the commitment register in the state
\begin{equation}
\rho_b=
\begin{cases}
\proj{z_J}_{\beta(x_0)} & \text{if } b=0,\\[1mm]
\proj{z_J}_{\beta(x_1)} & \text{if } b=1.
\end{cases}
\end{equation}

In the opening phase, Alice reveals the challenge $x_0$, together with $b$, $a=f_1(x_0)$, and $z=f_2(x_0)$.

In Step 9, Bob queries the locked HLPUF on $(x_0,\ket{\verst{a}})$. Since $a=f_1(x_0)$, the verifier state is exactly the correct state for the original challenge, and the token outputs the $l$-qubit BB84 state $\ket{\comst{z}}$. Hence Bob's test in Step 10 accepts with probability $1$.

In Step 11, Bob measures the stored commitment register in basis
\begin{equation}
\gamma_b=
\begin{cases}
\beta(x_0) & \text{if } b=0,\\
\beta(x_1) & \text{if } b=1.
\end{cases}
\end{equation}
This is exactly the basis in which Alice prepared the state $\rho_b$. Therefore the measurement outcome is precisely $z_J$, and Bob accepts. Thus the protocol opens to the correct bit and is accepted with probability 1.
\end{proof}

To prove hiding, we first prove a Lemma to show that, without any side information from the HLPUF, commitment states $\rho_0$ and $\rho_1$ are perfectly indistinguishable, and further the conditional states for the ordered pairs are also indistinguishable.

\begin{lemma}[Perfect intrinsic ensemble hiding]
\label{lem:intrinsic-hiding}
Fix a public ordered pair $(x_0,x_1)$ and its associated ordering $J$ produced by Algorithm~\ref{alg:fake}. Given Property (iv) in Proposition~\ref{prop:fake-properties}, the two honest commit ensemble states corresponding to $b=0$ and $b=1$ are perfectly indistinguishable:
\begin{equation}
    \dtr\!\left(\rho_{0},\rho_{1}\right)=0.
\end{equation}
\end{lemma}

\begin{proof}
Let $x_0$ denote the original challenge, and let $x_1$ denote the alternative element of the pair. These are sent labelled during the commitment stage, with Bob having full knowledge of their entire classical description. Given the properties of unforgeability of the HLPUF, he cannot know beyond a negligible probability what the corresponding $f_2(x_0)$ or $f_2(x_1)$ are. This consequently results in him gaining no information from using these challenges to either measure the states early, or to try to unlock the HLPUF before the reveal stage. This forces him to attempt to distinguish the two possible quantum ensembles with no knowledge of the underlying encoded string.

Let $U:=f_2(x_0)_J\in\bits^{\ell_{min}}$ denote the random payload string. The states relevant to Bob are the ensemble density operators
\begin{equation}
    \rho_0=\sum_{u\in\bits^{\ell_{min}}}\Pr[U=u]\,\proj{u}_{\beta(x_0)}
\end{equation}
and
\begin{equation}
    \rho_1=\sum_{u\in\bits^{\ell_{min}}}\Pr[U=u]\,\proj{u}_{\beta(x_1)}.
\end{equation}
By Property (iv) of Proposition~\ref{prop:fake-properties}, $U$ is uniformly distributed over $\bits^{\ell_{min}}$. Therefore
\begin{equation}
    \rho_0=\frac{1}{2^{\ell_{min}}}\sum_{u\in\bits^{\ell_{min}}}\proj{u}_{\beta(x_0)}=\frac{\Id}{2^{\ell_{min}}}
\end{equation}
and similarly
\begin{equation}
    \rho_1=\frac{1}{2^{\ell_{min}}}\sum_{u\in\bits^{\ell_{min}}}\proj{u}_{\beta(x_1)}=\frac{\Id}{2^{\ell_{min}}}.
\end{equation}
Hence $\rho_0=\rho_1$, and therefore
\begin{equation}
    \dtr\!\left(\rho_{0},\rho_{1}\right)=0.
\end{equation}
\end{proof}
Now we show that the hiding parameter is essentially bounded by the unforgeability of HLPUF, and hence we show an exponentially small hiding parameter.  
\begin{theorem}[Hiding]
\label{thm:hiding}
Let $\eps_{\mathrm{forge}}$ be the unforgeability parameter of HLPUF according to Assumption~\ref{ass:unf}. Then Protocol~\ref{bc-prot} is $\eps_{\mathrm{hide}}$-hiding with
\begin{equation}
\eps_{\mathrm{hide}}\le \eps_{\mathrm{forge}}.
\end{equation}
In particular, if the HLPUF is exponentially unforgeable in $s$, then $\eps_{\mathrm{hide}}=\negl(s).$
\end{theorem}

\begin{proof}
Fix an arbitrary cheating Bob. Let $F$ be the event that Bob successfully produces a verifier state accepted by the locked HLPUF on at least one of the two announced candidate challenges before the opening phase. By definition,
\begin{equation}
\Pr[F]\le \eps_{\mathrm{forge}}.
\end{equation}
Since $\eps_{\mathrm{forge}}$ isthe maximum probability, over all cheating strategies for Bob carried out before the opening phase, that Bob successfully produces a verifier state accepted by the locked HLPUF.

Let $\rho_0$ and $\rho_1$ denote Bob's full commit-phase states for honest commitments to $0$ and $1$, respectively. Decompose them according to the event $F$:
\begin{equation}
\rho_b
=
(1-\Pr[F])\,\rho_b^{\neg F}
+
\Pr[F]\,\rho_b^{F},
\qquad
b\in\bits.
\end{equation}
Conditioned on $\neg F$, Bob has no successful verifier forgery, so by Lemma~\ref{lem:intrinsic-hiding},
\begin{equation}
\rho_0^{\neg F}=\rho_1^{\neg F}.
\end{equation}
Therefore, using joint convexity of trace distance,
\begin{equation}
\dtr(\rho_0,\rho_1)
\le
(1-\Pr[F])\,\dtr(\rho_0^{\neg F},\rho_1^{\neg F})
+
\Pr[F]\,\dtr(\rho_0^{F},\rho_1^{F}).
\end{equation}
The first term vanishes and the second is at most $\Pr[F]$, since trace distance is always at most $1$. Hence
\begin{equation}
\dtr(\rho_0,\rho_1)\le \Pr[F]\le \eps_{\mathrm{forge}}.
\end{equation}
This is exactly the claimed bound and by Assumption~\ref{ass:unf}, the right hand side quantity is negligible.
\end{proof}
\begin{lemma}[Bounding Operator Norm of Opening Measurements]
\label{lem:sparse-verifier}
Fix the public pair $(x_0,x_1)$ and its corresponding ordered list $J$, with underlying set
$J=\{j\in[t]:(x_0)_j\neq(x_1)_j\}$ and $|J|=\ell_{min}$, and set $x=x_0$. Let
\begin{equation}
\mathcal{R}_x \subseteq \bits^s\times\bits^t
\end{equation}
denote the set of verifier/output pairs \((a,z)\) such that, if Bob queries the locked HLPUF on \((x,\ket{\verst{a}})\) and then measures the returned state with the test
\begin{equation}
\{\proj{\comst{z}},\,\Id-\proj{\comst{z}}\},
\end{equation}
the verification accepts. For this fixed $J$, let
\begin{equation}
    S:=\{z_J:\exists\;a\in\{0,1\}^s\;\text{such that}\; (a,z)\in\mathcal{R}_x\}.
\end{equation}
Here the operator argument below treats $\mathcal{R}_x$ as an acceptance relation containing at most one accepted output string $z$ for each verifier string $a$. The  finite-copy test above is interpreted probabilistically.
Let $H:=H^{\otimes\ell_{min}}$, and let,
\begin{subequations}
    \begin{equation}
        P:=\sum_{y\in S}\proj{y}_{\beta(x_0)}
    \end{equation}
    \begin{equation}
        Q:=\sum_{y\in S}\proj{y}_{\beta(x_1)}=HPH=\sum_{y\in S}H\proj{y}_{\beta(x_0)}H
    \end{equation}
\end{subequations}
the operator norm of $||P+Q||_{\infty}$ is bounded by,
\begin{equation}
    ||P+Q||_{\infty}\le1+2^{\frac{2s-\ell_{min}}{2}}
\end{equation}
\begin{proof}
The operator norm can be bounded by 
\begin{equation}
    ||P+Q||_{\infty}=1+||PQ||_{\infty}\le1+||PQ||_{F}=1+\sqrt{\Tr[(PQ)^{\dagger}(PQ)]}
\end{equation}
Where $||\cdot||_F$ denotes the Frobenius norm. Given that both $P$ and $Q$ are projectors and the cyclic property of the trace
\begin{equation}
    \Tr[(PQ)^{\dagger}(PQ)]=\Tr[QPPQ]=\Tr[QPQ]=\Tr[HPHPHPH]=\Tr[PHPHP]=\Tr[PHPH]
\end{equation}
As trace is the sum of basis vectors
\begin{equation}
    \Tr[A]=\sum_u\bra{u}A\ket{u}
\end{equation}
Expressing $P$ as a sum of states and expanding out the trace
\begin{equation}
\begin{split}
    &\Tr[\sum_{y,z\in S}\proj{y}H\proj{z}H]=\sum_{u\in\{0,1\}^{\ell_{min}}}\sum_{y,z\in S}\braket{u|y}\bra{y}H\proj{z}H\ket{u}=\\
    &\sum_{y,z\in S}\delta_{u,y}2^{-\ell_{min}}(-1)^{y\cdot z}(-1)^{u\cdot z}=\sum_{y,z\in S}2^{-\ell_{min}}(-1)^{2y\cdot z}=2^{-\ell_{min}}|S|^2
\end{split}
\end{equation}
Note that $|S|\le2^s$ under the acceptance-relation above, since there are $2^s$ input strings $a$ and at most one associated output string for each $a$. Taking the square root leads to
\begin{equation}
    ||P+Q||_{\infty}=1+||PQ||_{\infty}\le1+\sqrt{\frac{|S|^2}{2^{\ell_{min}}}}=1+2^{\frac{2s-\ell_{min}}{2}}
\end{equation}
\end{proof}
\end{lemma}
\begin{theorem}[Binding]
    \label{thm:binding} Protocol~\ref{bc-prot} is $\eps_{\mathrm{bind}}$-binding with
    \begin{equation}
        \eps_{\mathrm{bind}}\le 2^{\frac{2s-\ell_{min}}{2}}
    \end{equation}
    which results in
    \begin{equation}
        p_0+p_1\le 1+2^{\frac{2s-\ell_{min}}{2}}
    \end{equation}
\begin{proof}
    Fix an arbitrary cheating strategy for Alice. This strategy may be fully general: Alice may send Bob one half of an arbitrary entangled state in the commit phase, keep a private ancilla, measure that ancilla adaptively during the opening phase, and choose her classical opening data based on the measurement outcome. We now show that even in this most general setting, the total cheating probability is bounded by the claimed expression. 
    
    By quantum state purification, any mixed state can be represented as a pure bipartite state, in a larger Hilbert space. Alice can leverage this knowledge to keep an entangled register A, in which after concluding on her commitment, Alice can measure this state in a given way to influence Bob’s register B. Let $\rho_B$ denote Bob's reduced commitment state immediately before the opening phase. For an attempted opening of bit $b$, Alice's measurement on her private system induces a collection of subnormalised conditional states $\{\tau_{b,y}\}_{y\in S}$ on Bob's register, satisfying
    \begin{equation}
        \tau_{b,y}\succeq0,\qquad \sum_{y\in S}\tau_{b,y}=\rho_B.
    \end{equation}
    Let $\Pi_{0,y}:=\proj{y}_{\beta(x_0)}$ and $\Pi_{1,y}:=\proj{y}_{\beta(x_1)}$. Conditioned on opening data corresponding to $y$, Bob accepts the commitment register only through the projector $\Pi_{b,y}$. Hence
    \begin{equation}
        p_b\le\sum_{y\in S}\Tr[\Pi_{b,y}\tau_{b,y}].
    \end{equation}
    Since $\Pi_{0,y}\le P$ and $\Pi_{1,y}\le Q$ for every $y\in S$, it follows that
    \begin{equation}
        p_0\le\Tr[P\rho_B],\qquad p_1\le\Tr[Q\rho_B].
    \end{equation}
    Therefore
    \begin{equation}
        p_0+p_1\le\Tr[(P+Q)\rho_B]\le\norm{P+Q}_{\infty}.
    \end{equation}
    Applying Lemma~\ref{lem:sparse-verifier} gives
    \begin{equation}
        p_0+p_1\le||P+Q||_{\infty}\le1+2^{\frac{2s-\ell_{min}}{2}}.
    \end{equation}
\end{proof}
\end{theorem}

\section{Coin flipping from bit commitment}

We now build a coin-flipping protocol from Protocol~\ref{bc-prot} as a black box.

\begin{protocolboxenv}
\textbf{Protocol CF (coin flipping from BC Protocol~\ref{bc-prot}).}

\begin{enumerate}[label=\arabic*.]
    \item Alice samples a uniform random bit \(a\leftarrow\bits\) and commits to \(a\) using Protocol BC$^\star$.
    \item Bob samples a uniform random bit \(b\leftarrow\bits\) and sends \(b\) to Alice.
    \item Alice opens the commitment.
    \item If Bob rejects the opening, the protocol aborts. Otherwise both parties output
    \[
    c = a\oplus b.
    \]
\end{enumerate}
\end{protocolboxenv}

\begin{theorem}[Correctness of coin flipping]
If both parties are honest, Protocol CF$^\star$ outputs a uniform bit and does not abort.
\end{theorem}

\begin{proof}
By Theorem~\ref{thm:correctness}, the underlying commitment opens correctly with certainty. Conditional on acceptance, the protocol does not abort. Since \(a\) and \(b\) are independent uniform bits, \(a\oplus b\) is uniform on \(\bits\).
\end{proof}

\begin{theorem}[Fairness of coin flipping]
\label{thm:coinflip}
Assume Protocol~\ref{bc-prot} is \(\eps_{\mathrm{hide}}\)-hiding and \(\eps_{\mathrm{bind}}\)-binding. Then Protocol CF has bias at most
\begin{equation}
\delta_{\mathrm{CF}} \le \frac12\max\{\eps_{\mathrm{hide}},\eps_{\mathrm{bind}}\}.
\end{equation}
In particular, under Theorems~\ref{thm:hiding} and \ref{thm:binding},
\begin{equation}
    \delta_{\mathrm{CF}}\le\frac12\max\left\{\eps_{\mathrm{forge}},\,2^{-\frac{\ell_{\min}}{4}}\right\}.
\end{equation}
\end{theorem}

\begin{proof}
We consider dishonest Bob and dishonest Alice separately.

\emph{Dishonest Bob.} Bob receives the commitment transcript and the commitment register before sending his bit \(b\). By Theorem~\ref{thm:hiding}, his distinguishing advantage between commitments of \(0\) and \(1\) is at most \(\eps_{\mathrm{hide}}\), hence his bias is at most \(\eps_{\mathrm{hide}}/2\).

\emph{Dishonest Alice.} Alice commits before learning Bob's bit \(b\). To force a target outcome \(c^*\) after receiving \(b\), she must successfully open the commitment to the bit
\begin{equation}
a=c^*\oplus b.
\end{equation}
Since honest Bob chooses \(b\) uniformly at random,
\begin{equation}
\frac12 p_0+\frac12 p_1 \le \frac12(1+\eps_{\mathrm{bind}}),
\end{equation}
so her bias is at most \(\eps_{\mathrm{bind}}/2\). Taking the larger of the two cheating biases gives the claim.
\end{proof}

\section{Discussion and future work}

We have presented a hardware-based route to statistically secure bit commitment and coin flipping in a non-relativistic setting. Our construction uses Hybrid-Locked PUFs to benefit from both the availability of classical PUFs and the security advantages of quantum protocols. We then use this hardware token in a setting beyond its usual application, which is typically one where the parties are honest and security is only required against a network adversary.
\\
We design the protocol so that the commitment is tightly tied to the specific challenge-response behaviour of the hardware module, and we exploit the properties of the HLPUF to achieve statistical hiding and binding. The main challenge in designing such quantum protocols is to evade the impossibility of quantum bit commitment while at the same time forcing the commitment state to reflect genuine hardware behaviour. Our protocol addresses this by using the unforgeability property to obtain hiding, and the lock mechanism to enable a verification procedure that yields binding. We believe that both the protocol and the proof techniques developed here open new pathways toward secure two-party communication and computation in quantum networks using intrinsically weak, off-the-shelf PUFs whose security is amplified through protocol design.
\\
An appealing aspect of the protocol is that it allows for a variant in which early measurement is possible. Bob may measure each qubit in a random basis when received, recording each outcome. Then in the commit stage, he only keeps the information for which the bits of $\beta(x_b)$ match the bases he measured in. Although this weakens the security parameter, this allows for a statistically secure protocol in which no quantum memory is required. Naturally, a full analysis on this protocol would be for future work.
\\
In contrast to the bounded or noisy quantum storage models, our construction does not require the assumptions regarding an adversary's storage capabilities. This allows for the protocol to be indefinitely secure, rather than scaling either the number of qubits sent or time taken to run the protocol. The trade off involves the replacement of an adversary's memory with that of the hardness of forging the HLPUF. It remains an open question as to how the tamper proof-ness of the device would be implemented and scaled.
\\
There are several natural next steps. One direction is to design simpler variants of the protocol with a less complex fake-challenge generation procedure. On the proof side, it would be desirable to obtain a fully composable security treatment, since the present analysis is carried out in the stand-alone setting. It would also be interesting to understand how far the same design principles extend to stronger tasks such as string commitment, oblivious transfer, and other mistrustful multiparty functionalities. Finally, experimental implementation of these protocols would be an especially intriguing direction for future work.

\subsection*{Acknowledgments}
The authors would like to thanks Dominik Leichtle and Alexandru Cojocaru for helpful discussions at different stages of this work.  
RD acknowledges the support of the PhD Studentship from the Quantum Computing and Simulation QCS hub, grant number EP/T001062/1. MD acknowledges the support of the UK Engineering and Physical Sciences Research Council, the Integrated Quantum Networks Hub, grant reference EP/Z533208/1, and the Quantum Advantage Pathfinder (QAP), with grant reference EP/X026167/1.

\bibliography{arxiv-version/refs}

\appendix
\section{Complexity Analysis of Algorithm~\ref{alg:fake}}\label{app:alg-complexity}
Here we provide a more detailed formulation of the Theorems~\ref{thm:prob-suc-alg1}, ~\ref{thm:combinatorial-feasibility} and \ref{thm:fake-query}, and their proofs.\\
\begin{theorem}[Success Probability of Step 3 in  Algorithm~\ref{alg:fake}]
\label{thm:prob-suc-alg1-ext}
Let $t$ be the length of $f_2(x)$ and $\ell_{min} \leq t$ be the number of qubits in the register $\rho_b$. Let $\delta_r$ be the bias of the HLPUF. The probability that the algorithm succeeds in one run is:
\begin{equation}
    \Pr[success]=\sum_{y=0}^{\ell_{min}}\left(\frac{1}{2}\right)^{\ell_{min}}\binom{\ell_{min}}{y}\sum_{k=y}^{t-(\ell_{min}-y)}\binom{t}{k}\left(\frac{1}{2}+\delta_r\right)^k\left(\frac{1}{2}-\delta_r\right)^{t-k}
\end{equation}
\begin{proof}
    The algorithm succeeds whenever $f_2(x_0)$ contains at least as many $0s$ as $r$ and at least as many $1s$ as $r$. This occurs when the properties $\wt(f_2(x))\geq\wt(r)\;\;\wedge\;\;t-\wt(f_2(x))\geq l_{min}-\wt(r)$ are satisfied. The distribution of the string $f_2(x_0)$ can be modelled by the binomial distribution
    \begin{equation}
        X:=\text{wt}(f_2(x_0))\sim\text{Bin}(t,p)\;\;,\;\;p=\frac{1}{2}+\delta_r
    \end{equation}
    with expectation
    \begin{equation}
        \mu_x:=\mathbb{E}[X]=tp
    \end{equation}
    while the distribution of the string $r$ can be modelled by the binomial distribution
    \begin{equation}
        Y:=\text{wt}(r)\sim\text{Bin}(\ell_{min},\frac{1}{2})
    \end{equation}
    with expectation
    \begin{equation}
        \mu_y:=\mathbb{E}[Y]=\frac{\ell_{min}}{2}=\frac{\alpha t}{2}
    \end{equation}
    Conditioning on the event that $Y=y$, success occurs if $y\leq X\leq t-(\ell_{min}-y)$, therefore
    \begin{equation}
        \Pr[success|Y=y]=\sum_{k=y}^{t-(\ell_{min}-y)}\binom{t}{k}\left(p\right)^k\left(1-p\right)^{t-k}
    \end{equation}
    By the law of total probability
    \begin{equation}
        \Pr[success]=\sum_{y=0}^{\ell_{min}}\Pr[success|Y=y]\Pr[Y=y]
    \end{equation}
    Therefore:
    \begin{equation}
    \begin{split}
    &\Pr[success]=\Pr[\wt(f_2(x))\geq\wt(r)\;\;\wedge\;\;t-\wt(f_2(x))\geq l_{min}-\wt(r)]=\\
    &\sum_{y=0}^{\ell_{min}}\left(\frac{1}{2}\right)^{\ell_{min}}\binom{\ell_{min}}{y}\sum_{k=y}^{t-(\ell_{min}-y)}\binom{t}{k}\left(\frac{1}{2}+\delta_r\right)^k\left(\frac{1}{2}-\delta_r\right)^{t-k}
    \end{split}
    \end{equation}
\end{proof}
\end{theorem}

\begin{theorem}[High-probability combinatorial feasibility of Algorithm~\ref{alg:fake}]
\label{thm:combinatorial-feasibility-ext}
Let $x_0 \xleftarrow{\$} \{0,1\}^{n}$, and suppose that the bits of
$f_2(x_0)\in\{0,1\}^{t}$ are independently and identically distributed with a bias at most $\delta_r$. Let $r\xleftarrow{\$}\{0,1\}^{\ell_{\min}}$ be sampled independently,
where
\begin{equation}
    \ell_{\min}=\alpha t,\qquad 0<\alpha<1,
\end{equation}
and assume that \(\ell_{\min}\) is an integer. If
\begin{equation}
    \frac{\alpha}{2}<p<1-\frac{\alpha}{2},
\end{equation}
then the combinatorial feasibility condition in Step~3 of Algorithm~1,
namely
\begin{equation}
    \operatorname{wt}(f_2(x_0))\geq \operatorname{wt}(r)\quad\text{and}\quad t-\operatorname{wt}(f_2(x_0))\geq\ell_{\min}-\operatorname{wt}(r),
\end{equation}
is satisfied with probability
\begin{equation}
    \Pr[\mathrm{feasible}]
    =
    1-e^{-\Omega(t)}.
\end{equation}
Equivalently, except with probability \(e^{-\Omega(t)}\), Algorithm~1
can select \(\ell_{\min}\) distinct indices
\(j_1,\ldots,j_{\ell_{\min}}\in[t]\) such that
\begin{equation}
    f_2(x_0)_{j_i}=r_i\qquad
    \text{for every }i\in[\ell_{\min}].
\end{equation}
The constant implicit in the \(\Omega(t)\) notation may depend on
\(\alpha\) and \(p\), but not on \(t\).
\begin{proof}
Let
\begin{equation}
    X:=\operatorname{wt}(f_2(x_0))
    \sim \operatorname{Bin}(t,p),
\end{equation}
with expectation
\begin{equation}
    \mu_x:=\mathbb{E}[X]=tp,
\end{equation}
and let
\begin{equation}
    Y:=\operatorname{wt}(r)\sim \operatorname{Bin}\left(\ell_{\min},\frac{1}{2}\right),
\end{equation}
with expectation
\begin{equation}
    \mu_y:=\mathbb{E}[Y]=\frac{\ell_{\min}}{2}=\frac{\alpha t}{2}.
\end{equation}
Since $x_0$ and $r$ are sampled independently, the random variables $X$ and $Y$ are independent.\\
The combinatorial feasibility condition in Step~3 is satisfied precisely when
\begin{equation}
    X\geq Y\quad\text{and}\quad 
    t-X\geq \ell_{\min}-Y,
\end{equation}
or equivalently,
\begin{equation}
    Y\leq X\leq t-\ell_{\min}+Y.
\end{equation}
Since
\begin{equation}
    \frac{\alpha}{2}<p<1-\frac{\alpha}{2},
\end{equation}
we may choose a constant $\varepsilon_y\in(0,1)$ sufficiently small such that
\begin{equation}
    \frac{\alpha}{2}(1+\varepsilon_y)<p<1-\frac{\alpha}{2}(1+\varepsilon_y).
\end{equation}
Define
\begin{equation}
    y_{\max}:=(1+\varepsilon_y)\mu_y,\qquad
    y_{\min}:=(1-\varepsilon_y)\mu_y.
\end{equation}
Applying multiplicative Chernoff bounds to $Y$ gives
\begin{equation}
    \Pr\!\left[Y\geq (1+\varepsilon_y)\mu_y\right]\leq\exp\left(-\frac{\varepsilon_y^2\mu_y}{2+\varepsilon_y}\right)
\end{equation}
and
\begin{equation}
    \Pr\!\left[Y\leq (1-\varepsilon_y)\mu_y\right]\leq\exp\left(-\frac{\varepsilon_y^2\mu_y}{2}\right).
\end{equation}
Since $\mu_y=\alpha t/2$, it follows that the event
\begin{equation}
    \mathcal{E}_Y:=\left\{y_{\min}\leq Y\leq y_{\max}\right\}
\end{equation}
occurs with probability
\begin{equation}
    \Pr[\mathcal{E}_Y]=1-e^{-\Omega(t)}.
\end{equation}
Conditioned on $\mathcal{E}_Y$, it is sufficient that
\begin{equation}
    y_{\max}\leq X\leq t-\ell_{\min}+y_{\min},
\end{equation}
because then
\begin{equation}
    Y\leq y_{\max}\leq X
\end{equation}
and
\begin{equation}
    X\leq t-\ell_{\min}+y_{\min}\leq t-\ell_{\min}+Y.
\end{equation}
Define the event
\begin{equation}
    \mathcal{E}_X:=\left\{y_{\max}\leq X\leq t-\ell_{\min}+y_{\min}\right\}.
\end{equation}
By the union bound,
\begin{equation}
    \Pr[\mathcal{E}_X]\geq1-\Pr[X<y_{\max}]-\Pr[X>t-\ell_{\min}+y_{\min}].
\end{equation}
By the choice of $\varepsilon_y$,
\begin{equation}
    y_{\max}=\frac{\alpha t}{2}(1+\varepsilon_y)<tp=\mu_x.
\end{equation}
Therefore, the lower-tail Chernoff bound gives
\begin{equation}
    \Pr[X<y_{\max}]\leq\exp\left[-\frac{\left(1-\frac{y_{\max}}{\mu_x}\right)^2\mu_x}{2}\right].
\end{equation}
Substituting $\mu_x=tp$ and
$y_{\max}=\alpha t(1+\varepsilon_y)/2$, we obtain
\begin{equation}
    \Pr[X<y_{\max}]\leq\exp\left[-\frac{t\left(p-\frac{\alpha}{2}(1+\varepsilon_y)\right)^2}{2p}\right].
\end{equation}
Similarly,
\begin{equation}
    t-\ell_{\min}+y_{\min}=t\left(1-\frac{\alpha}{2}(1+\varepsilon_y)\right)>tp=\mu_x.
\end{equation}
Using the upper-tail Chernoff bound
\begin{equation}
    \Pr[X>(1+\delta)\mu_x]\leq\exp\left(-\frac{\delta^2\mu_x}{2+\delta}\right),
\end{equation}
with
\begin{equation}
    \delta=\frac{t-\ell_{\min}+y_{\min}}{\mu_x}-1,
\end{equation}
gives
\begin{equation}
    \Pr[X>t-\ell_{\min}+y_{\min}]\leq\exp\left[-\frac{\left(t-\ell_{\min}+y_{\min}-\mu_x\right)^2}{\mu_x+t-\ell_{\min}+y_{\min}}\right].
\end{equation}
Substituting the corresponding values yields
\begin{equation}
    \Pr[X>t-\ell_{\min}+y_{\min}]\leq\exp\left[-\frac{t\left(1-p-\frac{\alpha}{2}(1+\varepsilon_y)\right)^2}{1+p-\frac{\alpha}{2}(1+\varepsilon_y)}\right].
\end{equation}
Consequently,
\begin{equation}
\begin{aligned}
    \Pr[\mathcal{E}_X]\geq\;&1-\exp\left[-\frac{t\left(p-\frac{\alpha}{2}(1+\varepsilon_y)\right)^2}{2p}\right]\\
    &-\exp\left[-\frac{t\left(1-p-\frac{\alpha}{2}(1+\varepsilon_y)\right)^2}{1+p-\frac{\alpha}{2}(1+\varepsilon_y)}\right].
\end{aligned}
\end{equation}
Both exponents are linear in $t$ with strictly positive coefficients. Hence,
\begin{equation}
    \Pr[\mathcal{E}_X]=1-e^{-\Omega(t)}.
\end{equation}
Finally, the event $\mathcal{E}_X\cap\mathcal{E}_Y$ implies
\begin{equation}
    Y\leq X\leq t-\ell_{\min}+Y.
\end{equation}
Since $\mathcal{E}_X$ depends only on $X$, while $\mathcal{E}_Y$ depends only on $Y$, independence gives
\begin{equation}
\begin{aligned}
    \Pr\!\left[Y\leq X\leq t-\ell_{\min}+Y\right]
    &\geq\Pr[\mathcal{E}_X\cap\mathcal{E}_Y]\\
    &=\Pr[\mathcal{E}_X]\Pr[\mathcal{E}_Y]\\
    &=\left(1-e^{-\Omega(t)}\right)\left(1-e^{-\Omega(t)}\right)\\
    &=1-e^{-\Omega(t)}.
\end{aligned}
\end{equation}
This proves the claimed high-probability combinatorial feasibility.
\end{proof}
\end{theorem}
\begin{theorem}[HLPUF query complexity of Algorithm~\ref{alg:fake}]
\label{thm:fake-query-ext}
Assume the setting of Theorem~\ref{thm:prob-suc-alg1-ext}. Let $Q$ denote the total number of unlocked-HLPUF queries made by Algorithm~\ref{alg:fake} on input $x_0$, counting the initial query to obtain $f(x_0)$. Conditioned on the combinatorial feasibility test being satisfied, then
\begin{equation}
    \mathbb{E}[Q]
    \le
    1+\frac{1}{p_{\mathrm{sep}}},
\end{equation}
where
\begin{equation}
    p_{\mathrm{sep}}
    \ge
    1-\left(1-\left(\frac12-\delta_r\right)^2\right)^{s/2}.
\end{equation}
Equivalently,
\begin{equation}
    \mathbb{E}[Q]
    \le
    1+\frac{1}{
    1-\left(1-\left(\frac12-\delta_r\right)^2\right)^{s/2}
    }.
\end{equation}
In particular, for every fixed $\delta_r<\frac{1}{2}$,
\begin{equation}
    \mathbb{E}[Q]=2+e^{-\Omega(s)}.
\end{equation}
\end{theorem}

\begin{proof}
Algorithm~\ref{alg:fake} always makes one unlocked-HLPUF query on the original challenge $x_0$ in order to obtain
\begin{equation}
    f(x_0)=f_1(x_0)\,\|\,f_2(x_0).
\end{equation}
After that, the balancing filter uses only classical computation on $x_0$, $f_2(x_0)$, and the public map $\beta$, so it requires no additional HLPUF queries. Every additional unlocked-HLPUF query is used only to test whether a candidate challenge $x$ satisfies the verifier-separation condition.

Write $s=2k$, so the verifier state has $k=\frac{s}{2}$ BB84 qubits. Fix one verifier qubit of the original challenge, determined by two bits of $f_1(x_0)$. For a tested candidate challenge $x$, the corresponding verifier qubit is orthogonal to the original one whenever the basis bit matches and the value bit is opposite. By the independent bitwise bias assumption, the probability of this orthogonal event is at least
\begin{equation}
    q_{\perp}\ge \left(\frac12-\delta_r\right)^2.
\end{equation}
If this orthogonal event occurs on at least one verifier qubit, then the total verifier overlap is zero, and in particular
\begin{equation}
    \abs{\mbraket{\verst{f_1(x_0)}}{\verst{f_1(x)}}}^2\le 2^{-s/2}.
\end{equation}
Therefore a sufficient condition for a candidate challenge to pass the verifier-separation test is that at least one of the $k=\frac{s}{2}$ verifier qubits be orthogonal. Since the qubits are independent,
\begin{equation}
    p_{\mathrm{sep}}\ge1-(1-q_{\perp})^{k}\ge1-\left(1-\left(\frac12-\delta_r\right)^2\right)^{s/2}.
\end{equation}
Under the sampling assumption, each tested candidate succeeds independently with probability at least $p_{\mathrm{sep}}$. Thus the number $N$ of additional candidate tests before the first successful one is stochastically dominated by a geometric random variable with success probability $p_{\mathrm{sep}}$, and
\begin{equation}
    \mathbb{E}[N]\le \frac{1}{p_{\mathrm{sep}}}.
\end{equation}
Since $Q=1+N$, we obtain
\begin{equation}
    \mathbb{E}[Q]\le 1+\frac{1}{p_{\mathrm{sep}}}.
\end{equation}
Finally, for every fixed \(\delta_r<1/2\), the quantity
\begin{equation}
    1-\left(\frac12-\delta_r\right)^2
\end{equation}
is a constant strictly smaller than $1$. Hence
\begin{equation}
    1-p_{\mathrm{sep}}\le\left(1-\left(\frac12-\delta_r\right)^2\right)^{s/2}=e^{-\Omega(s)},
\end{equation}
which implies
\begin{equation}
    \mathbb{E}[Q]=1+\frac{1}{1-e^{-\Omega(s)}}=2+e^{-\Omega(s)}.
\end{equation}
\end{proof}
\end{document}